\documentclass{llncs}
\usepackage[utf8]{inputenc}
\usepackage{t1enc,color,array,rotating,graphicx,mathtools}
\usepackage{soul}
\usepackage{latexsym,amssymb,amsmath,amsfonts,stmaryrd,algorithm,algpseudocode}
\usepackage[english]{babel}
\usepackage{hyperref}
\numberwithin{equation}{section}
\usepackage{graphicx,nicefrac,wrapfig}
\usepackage[textsize=footnotesize,disable]{todonotes}

\usepackage{cleveref,multirow,tikz,xspace,comment,enumitem}
\usetikzlibrary{decorations,arrows,arrows.meta,shapes,patterns,automata,calc}
\usepackage{proof}
\usepackage{amssymb}
\usepackage{amstext}
\usepackage{verbatim}

\crefformat{footnote}{#2\footnotemark[#1]#3}

\newcommand{\KK}[1]{\todo[color=purple!20]{KK: #1}}
\newcommand{\NM}[1]{\todo[color=yellow!40]{NM: #1}}

\newcommand{\R}{\mathbb{R}}
\newcommand{\Q}{\mathbb{Q}}
\newcommand{\Pred}{\mathcal{P}}
\newcommand{\lin}{\mathit{lin}}
\newcommand{\nonlin}{\mathit{nl}}
\newcommand{\Fun}{\mathcal{F}}
\newcommand{\D}[1]{\emph{#1}}

\newcommand{\nl}{\mathit{nl}}

\newcommand\Title{A CDCL-style calculus for solving non-linear constraints}

\hypersetup{%
    pdftitle={\Title},%
    pdfauthor={F. Brauße, K. Korovin, M. Korovina, N. Müller},%
}

\title{\Title
\thanks{
The research leading to these results has received funding from the DFG grant WERA MU 1801/5-1 and the DFG/RFBR grant
CAVER BE 1267/14-1 and 14-01-91334.
 \protect\includegraphics[width=1em]{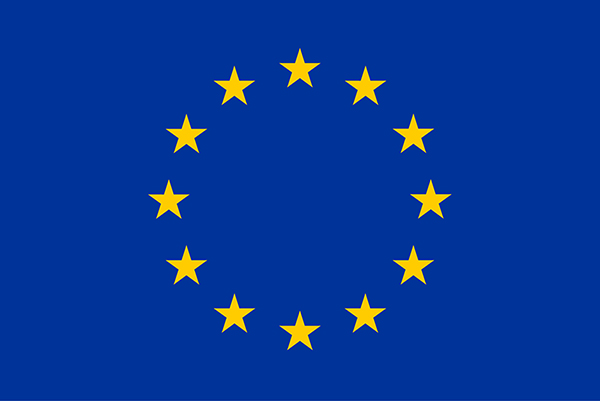} This project has received funding from the European Union’s Horizon 2020 research and innovation programme under the
Marie Skłodowska-Curie grant agreement No 731143.
}}

\author{F. Brauße\inst{1} \and K. Korovin \inst{2} \and M. Korovina\inst{3} \and N. Müller\inst{1} }
\institute{%
Abteilung Informatikwissenschaften, Universität Trier, Germany \and
The University of Manchester, UK \and
A.P. Ershov Institute of Informatics Systems, Novosibirsk, Russia}

\newcommand*\dom{\operatorname{dom}\nolimits}

\newcommand*\ksmt{\texttt{ksmt}\xspace}
\newcommand*\SAT{\texttt{sat}\xspace}
\newcommand*\UNSAT{\texttt{unsat}\xspace}
\newcommand*\UNKNOWN{\texttt{unknown}\xspace}

\newcommand*{\cons}{\mathbin{::}}
\newcommand*{\Res}{\operatorname{res}\nolimits}
\newcommand{\NLin}{\mathcal N}
\newcommand{\Lin}{\mathcal L}

\newcommand*\todofb[2][]{\todo[color=cyan!30,#1]{FB: #2}}
\newcommand*\True{\textsf{true}\xspace}
\newcommand*\False{\textsf{false}\xspace}
\newcommand*\nil{\textsf{nil}\xspace}
\newcommand*\notfalse{conflict-free\xspace}

\newcommand*\linnotfalse{linearly \notfalse}

\newcommand*\transition{\Rightarrow}
\newcommand*\Fdec{\ensuremath{\mathcal F_\mathrm{DA}}}
\newcommand*\MCSAT{\texttt{MCSAT}\xspace}
\newcommand*\NLSAT{\texttt{NLSAT}\xspace}
\newcommand*\dReal{\texttt{dReal}\xspace}

\begin{document}

\maketitle

\begin{abstract}
  In this paper we propose a novel approach for checking satisfiability of non-linear constraints over the reals, called \ksmt.
The procedure is based on conflict resolution in CDCL-style calculus, using a composition of symbolical and numerical methods. To deal with the non-linear components 
in case of conflicts we use numerically constructed restricted linearisations. This approach covers
a large number of computable non-linear real functions such as polynomials, rational or trigonometrical functions and beyond.
A prototypical implementation has been evaluated on several non-linear SMT-LIB examples and the results have been compared with state-of-the-art SMT solvers.
\end{abstract}

\section{Introduction}

Continuous constraints occur  naturally
in many areas of computer science such as verification of safety-critical systems, program analysis and theorem proving.
Historically, there have been two major approaches to solving continuous constraints. 
One of them is the symbolic approach, originated by the Tarski's decision procedure for the real closed fields \cite{Tarski}, and developed further in  procedures based on cylindrical decomposition (CAD)~\cite{C75}, Gr{\"o}bner basis~\cite{DBLP:journals/cca/Buchberger76,Kapur:2010:NAC}, and virtual substitution~\cite{LoosWeispfenning:93a,DolzmannSturm:97a}. 
%
Another one is the numerical approach, based on interval computations, where
the technique of interval constraint propagations have been explored to deal
with continuous constraints on compact intervals,
e.g., \cite{BenhamouG06,FHTRS07,Ratschan14,GaoAC12_1}.
It is well known that both approaches have their strengths and weaknesses
concerning completeness, efficiency and expressiveness.

Recently, a number of methods has been developed aimed at merging strengths of symbolical and numerical methods, e.g.~\cite{PassmorePM12,fontaine:hal-01946733,DBLP:journals/tocl/CimattiGIRS18,GCAI2016:AVATAR_Modulo_Theories}.
 In particular,  the approach developed in this paper is motivated by extensions of CDCL-style reasoning into domains beyond propositional logic such as linear~\cite{CR:KTV09,McMillanEtAl:09,KorovinTV11,KKS14} and polynomial constraints~\cite{JovanovicM12}. 
In this paper we develop a conflict-driven framework called \ksmt for solving non-linear constraints over large class of functions including polynomial, exponential and trigonometric functions. 
Our approach combines model guided  search for a satisfying solution and 
constraint learning as the result of failed attempts to extend the candidate solution. 

In the nutshell, our \ksmt algorithm works as follows. 
Given a set of non-linear constraints, we first separate the set into linear and non-linear parts.  
Then we incrementally extend a candidate solution into a solution of the system 
and when such extension fails we resolve the conflict by 
generating a lemma that excludes a region which includes the falsifying assignment. 
There are two types of conflicts: between linear constraints which are 
resolved in a similar way as in~\cite{CR:KTV09} and non-linear conflicts which are resolved by local linearisations developed in this paper. 
One of the important properties of our algorithm is that all generated lemmas are linear and hence the non-linear part of the problem remains unchanged during the search. In other words, our algorithm can be seen as applying gradual linear approximations of non-linear constraints by local linearisations guided by solution search in the CDCL-style. 

The quantifier-free theory of reals with transcendental functions is well known to be undecidable~\cite{DBLP:journals/jsyml/Richardson68}
and already problems with few variables pose considerable challenge for automated systems.
In this paper we focus on a practical algorithm for solving non-linear constraints applicable to problems with large number of variables rather than on completeness results.  
Our \ksmt algorithm can be used for both finding a solution and proving that no solution exist. 
In addition to a general framework we discuss how our algorithm works in a number of important cases such as polynomials, transcendental and some discontinuous functions.
In this paper we combine solution guided search in the style of conflict resolution,
bound propagation 
and MCSAT~\cite{MouraJ13} 
with linearisations of real computable functions.
The theory of computable functions has been developed in Computable
Analysis~\cite{We00} with implementations provided by exact real arithmetic~\cite{Muller00}.
Linearisations have been employed in different SMT theories before,
including NRA and a recently considered one with transcendental functions~\cite{DBLP:conf/vmcai/MarechalFKMP16,RTJ,DBLP:journals/tocl/CimattiGIRS18},
however, not for the broad class we consider here.
We define a general class of functions called {\it functions with decidable rational approximations} to which our approach is applicable. This class includes common transcendental functions, exponentials, logarithms but also some discontinuous functions.

We implemented the \ksmt algorithm and \todo{not really?} evaluated it on SMT benchmarks.
Our implementation is at an early stage and lacking many features but already 
outperforms many state-of-the-art SMT solvers on certain classes of problems.






\section{Preliminaries}

We consider the \todofb{`complete set of'} reals extended with non-linear functions $\R_{\nl}= (\R,\langle\Fun_{\lin}\cup \Fun_{\nonlin},\Pred\rangle)$,
where $\Fun_\lin$ 
consists of rational constants, addition  and multiplication by rational constants; 
$\Fun_\nonlin$ consists of a \todofb{detailed in \Cref{def:F}}
selection of non-linear functions including
multiplication, trigonometric, exponential and logarithmic functions;
$\Pred= \{{<},{\leq},{>},{\geq}\}$ 
are predicates.

We consider a set of variables $V$. We will use $x,y,z$ possibly with indexes for variables in $V$,
similar we will use \todofb{let's check $q,a,b,c$} $q,a,b,c,d$ for rationals,
$f,g$ for non-linear functions in $\Fun_{\nonlin}$.
Terms, predicates and formulas over $X$ are defined in the standard way. 
%
We will also use predicates ${\neq},=$, which can be defined using predicates in $\Pred$. 
An atomic formula is a formula of the form $t\diamond 0$ where $\diamond\in\Pred$. 
A literal is either an atomic formula or its negation. 
In this paper we consider only quantifier-free formulas in conjunctive normal form. 
We will use  conjunctions and sets of formulas interchangeably. 

We assume that terms are suitably normalised. 
A linear term is a term of the form $q_1x_1+\ldots + q_nx_n + q_0$.
A linear inequality is an atomic formula of the form $q_1x_1+\ldots + q_nx_n + q_0\diamond 0$. 
A linear clause is a disjunction of linear inequalities and a formula is in linear CNF if it is a conjunction of linear clauses.

%

\subsection{Separated linear form}
In this paper we consider the satisfiability problem of quantifier-free formulas in CNF over $\R_{\nl}$,
where the linear part is separated from the non-linear part which we call separated linear form.

\begin{definition}\label{def:separated}
A formula $F$ is in \D{separated linear form} if it is of the form $F=\Lin\cup \NLin$ where $\Lin$
is a set of clauses containing predicates only over linear terms and $\NLin$ is a set of unit-clauses each containing only non-linear literals of the form $x \diamond f(\vec t)$, where $f\in \Fun_{\nonlin}$, $\vec t$ is a vector of terms and $\diamond\in\Pred$.
\end{definition}

\begin{lemma}[Monotonic flattening]\label{lem:preproc}
Any quantifier-free formula $F$ in CNF over $\R_\nl$ can be transformed into an equi-satisfiable separated linear form in polynomial time.
\end{lemma}
\begin{proof}
Consider a clause $C$ in $F$ which contains a linear combination of non-linear terms, i.e., is of the form $C = qf(\vec t) + p\diamond 0 \vee D$, where  $f\in \Fun_{\nonlin}$ and $q\not = 0$. Then we introduce a fresh variable $x$, add 
$x\diamond' f(\vec t)$ into $S$ and replace $C$ with $qx + p\diamond 0 \vee D$. 
Here, $\diamond'$ is $\geq$, if either $q >0 $ and $\diamond \in \{ \leq, <\}$ or $q< 0$ and  $\diamond \in \{ \geq, >\}$; and $\diamond'$ is $\leq$ otherwise. 
The resulting formula is equi-satisfiable to $F$. The claim follows by induction on the non-linear monomials.
\end{proof}
Let us remark that monotonic flattening avoids introducing equality predicates, which is based on the monotonicity of linear functions. 
In some cases we need to flatten non-linear terms further (in particular to be able to represent terms as functions in the $\Fdec$ class introduced in \Cref{sec:impl}). In most cases this
can be done in the same way as in \Cref{lem:preproc} based on monotonicity of functions in corresponding arguments, but we may need to introduce linear conditions expressing regions of monotonicity. For simplicity of the exposition we will not consider such cases~here.



\subsection{Trails and Assignments}

Any  sequence of single variable assignments $\alpha\subset(V\times\mathbb Q)^*$ such that a variable is assigned at most once is called a \emph{trail}. By ignoring the order of assignments in $\alpha$, we will regard
$\alpha$ as a (partial) assignment of the real variables in $V$ and use $V(\alpha)\subseteq V$ to denote the set of variables assigned in $\alpha$. We use the notation
$\llbracket t\rrbracket^\alpha$ to\pagebreak[2] denote the (partial) application of $\alpha$ to a term $t$,
that is, the term resulting from replacing every free variable $x$ in $t$ such
that $x\in V(\alpha)$ by $\alpha(x)$ and evaluating term operations on constants in their domains.
We extend $\llbracket\cdot\rrbracket^\alpha$ to predicates over terms and to CNF in the usual way.
An evaluation of a formula results in $\True$ or $\False$, if all variables in the formula are assigned, 
or else in a partially evaluated formula. 
A \emph{solution} to a CNF $\mathcal C$ is a total assignment $\alpha$ such that
\todofb{`terms in $l$' as well as $\llbracket t\rrbracket^\alpha$ are defined, how to formulate better?}
each term in $\mathcal C$ is defined under $\alpha$ and
for each clause $C\in\mathcal C$ there is (at least) one literal $l\in C$ with $\llbracket l\rrbracket^\alpha=\True$.

Any triple $(\alpha,\Lin, \NLin)$ when $\alpha$ is a trail,
$\Lin$ is a set of clauses over linear predicates and $\NLin$ 
is a set of
unit clauses over non-linear predicates is called \emph{state}.
A state is called \emph{\linnotfalse} if
$\llbracket\Lin\rrbracket^\alpha\neq\False$.
It is called \emph{\notfalse} if it is \linnotfalse
and $\llbracket\NLin\rrbracket^\alpha\neq\False$. 

The main problem we consider in this paper is finding a solution to
$\Lin\land\NLin$ or showing that no
solution exists. 

\section{The \ksmt algorithm}

Our \ksmt algorithm will be based on a CDCL-type calculus~\cite{DBLP:journals/tc/Marques-SilvaS99,DBLP:journals/jacm/NieuwenhuisOT06} and is in the spirit of Conflict Resolution~\cite{CR:KTV09,KorovinTV11}, Bound Propagation~\cite{KV11:BP,DBLP:conf/synasc/DraganKKV13}, GDPLL~\cite{McMillanEtAl:09}, MCSAT~\cite{MouraJ13} and related algorithms.

The \ksmt calculus will be presented as a set of transition rules that operate on the states introduced previously.  
The \D{initial state} is a state of the form $(\nil,\Lin, \NLin)$. A final state will be reached when no further \ksmt transition rules (defined below) are applicable.

Informally, the \ksmt algorithm starts with a formula in separated linear form and the empty trail, and extends the trail until the solution is found or a trivial inconsistency is derived by applying the \ksmt transition rules. 
During the extension process the algorithm may encounter conflicts which are resolved by deriving lemmas which will be linear clauses. These lemmas are either derived by resolution between two linear clauses or by linearisation of non-linear conflicts, which is described in detail in \Cref{sec:nlin}. One of the important properties of our calculus is that we only generate linear lemmas
during the run of the algorithm and the non-linear part $\NLin$ remains fixed.

\subsection{General procedure}
Let $(\alpha,\Lin,\NLin)$ be a \notfalse state and
$z\in V\setminus V(\alpha)$ be a variable unassigned in $\alpha$.
Assume there is no $q\in\mathbb Q$ such that $(\alpha\cons z\mapsto q,\Lin,\NLin)$ is \linnotfalse.
That means that for any potential assignment $q$ there is a clause $D\in\Lin$ not satisfied under $\alpha\cons z\mapsto q$. Another way of viewing this
situation, called a \emph{conflict}, is that there are clauses
consisting under $\alpha$ only of predicates linear in and only depending on $z$ that
contradict each other. Analogously to resolution in propositional
logic,
\[ \begin{array}{c}
	A\vee\ell\quad B\vee\neg\ell \\ \hline
	A\vee B
\end{array} \]
the following inference rule we call \emph{arithmetical resolution on $x$} is sound \cite{CR:KTV09,McMillanEtAl:09} on clauses over linear predicates:
\[
	\begin{array}{c}
	A\vee(cx+d\leq 0)\quad B\vee(-c'x+d'\leq0)
	\\ \hline
	A\vee B\vee (c'd+cd'\leq0)
	\end{array}
\]
where $c, c'$ are positive rational constants and $d,d'$ are linear terms.
Similar rules exist for strict comparisons. 
%
We denote by $R_{\alpha,\Lin,z}$ a set of resolvents of clauses in  $\Lin$ upon variable $z$ such that 
$\llbracket R_{\alpha,\Lin,z}\rrbracket^\alpha =\False$. In Section~\ref{sec:check-bounds} 
we discuss how to obtain such a set.

We consider the following rules for transforming states into states under some preconditions, i.e., the binary relation $\transition$ on states.
\begin{description}
\item[Assignment refinement:]
    In order to refine an existing partial assignment $\alpha$ by assigning $z\in V$
    to $q\in\mathbb Q$ in a state $(\alpha,\Lin,\NLin)$, the state needs to be \linnotfalse,
    that is, no clause over linear predicates in $\Lin$ must be false under $\alpha$.
Additionally, under this assignment the clauses over linear predicates in $\Lin$
    must be valid under the new assignment, formally: For any state
    $(\alpha,\Lin,\NLin)$, $z\in V$ and $q\in\mathbb Q$
    \begin{equation}\small
    	(\alpha,\Lin,\NLin)\transition (\alpha\cons z\mapsto q,\Lin,\NLin)
    	\label{rule:extend} \tag{A}
    \end{equation}
    whenever
    $\llbracket\Lin\rrbracket^\alpha\neq\False$,
    $z\notin V(\alpha)$, and
    $\llbracket\Lin\rrbracket^{\alpha\cons z\mapsto q}\neq\False$.
    In the linear setting of \cite{CR:KTV09}, this rule
    exactly corresponds to ``assignment refinement''.
\item[Conflict resolution:]
    Assume despite state $(\alpha,\Lin,\NLin)$ being \linnotfalse and $z\in V$ unassigned in
    $\alpha$ there is no rational value to assign to $z$ that makes the resulting
    state \linnotfalse. This means, that for any $q\in\mathbb Q$ there is a
    \emph{conflict}, i.e., a clause in $\Lin$ that is false under $\alpha\cons z\mapsto q$.
    In order to progress in determining \SAT or \UNSAT, the partial assignment
    $\alpha$ needs to be excluded from the search space. Arithmetical resolution
    $R_{\alpha,\Lin,z}$ provides exactly that: a set of clauses preventing any $\beta\supseteq\alpha$ from being \linnotfalse. For any state $(\alpha,\Lin,\NLin)$ and $z\in V$
    \begin{equation}\small
        (\alpha,\Lin,\NLin)\transition(\alpha,\Lin\cup R_{\alpha,\Lin,z},\NLin)
        \label{rule:resolve} \tag{R}
    \end{equation}
    whenever $\llbracket\Lin\rrbracket^\alpha\neq\False$,
    $z\notin V(\alpha)$ and
    $\forall q\in\mathbb Q:\llbracket\Lin\rrbracket^{\alpha\cons z\mapsto q}=\False$.
    In the linear setting of \cite{CR:KTV09}, this rule
    corresponds to 
    ``conflict resolution''.

\item[Backjumping:]
    In case the state $(\alpha,\Lin,\NLin)$ contains one or more top-level assignments that make it not \linnotfalse, these assignments are removed.
    This is commonly known as backjumping. Indeed, when transitioning to applying this rule, the information on the size of the suffix of assignments to remove is already available, as is
    detailed in \Cref{info:backtracking-eq-backjumping}. Formally, for a state $(\alpha,\Lin,\NLin)$ such that $\llbracket\Lin\rrbracket^{\alpha}=\False$, let $\gamma$ be the maximal prefix of $\alpha$ such that  $\llbracket\Lin\rrbracket^\gamma\neq\False$.
    Then, Backjumping is defined as follows:
    \begin{equation}
        (\alpha,\Lin,\NLin)
        \transition
        (\gamma,\Lin,\NLin)
        \label{rule:backtrack} \tag{B}
    \end{equation}

\item[Linearisation:]
    The above rules are only concerned with keeping the (partial) assignment
    \linnotfalse. This rule extends the calculus to ensure that the non-linear clauses in $\NLin$ are \notfalse as well.
    In essence, the variables involved
    in a non-linear conflict are ``lifted'' into the linear domain by a
    linearisation of the conflict local to $\alpha$. The resulting state will not
    be \linnotfalse as is shown in \Cref{lem:linearization-effect}. Formally,
    if $(\alpha,\Lin,\NLin)$ is a state and $L_{\alpha,\NLin}$ a non-empty set of linearisation
    clauses
    as detailed in \Cref{sec:nlin}, then the rule reads as
    \begin{equation}\small
        (\alpha,\Lin,\NLin)\transition (\alpha,\Lin\cup L_{\alpha,\NLin},\NLin)
    	\label{rule:linearize} \tag{L}
    \end{equation}
    whenever $\llbracket\Lin\rrbracket^\alpha\neq\False$ and
    $\llbracket\NLin\rrbracket^\alpha=\False$.
\end{description}
Let us note that the set $\NLin$ remains unchanged over any sequence of states
obtained by successive application of the above rules.

\begin{lemma}[Soundness]\label{lem:coherence}
Let $I$ be an input instance in separated linear form.
Let $(S_0,S_1,\ldots,S_n)$ be a sequence of states $S_i=(\alpha_i,\Lin_i,\NLin)$ where $S_0$
is the initial state and each $S_{i+1}$ is derived from $S_i$ by application of
one of the rules \eqref{rule:extend}, \eqref{rule:resolve}, \eqref{rule:backtrack},
\eqref{rule:linearize}.
\begin{enumerate}
\item\label{it:single-step-sound}
    For all $i<n$ and total assignments $\alpha:V\to\mathbb Q$:
    $\llbracket\Lin_i\land\NLin\rrbracket^{\alpha}=\llbracket\Lin_{i+1}\land \NLin\rrbracket^{\alpha}$.
\item
    If no rule is applicable to $S_n$ then the following are equivalent:
    \begin{itemize}
    \item $I$ is satisfiable,
    \item $\alpha_n$ is a solution to $I$,
    \item $S_n$ is \linnotfalse,
    \item the trivial conflict clause $(1\leq0)$ is not in $\Lin_n$.
    \end{itemize}
\end{enumerate}
\end{lemma}
\begin{lemma}[Progress]\label{lem:reduce-search-space}%
Let $(S_i)_i$ be a sequence of states $S_i=(\alpha_i,\Lin_i,\NLin)$ produced from initial state $S_0$ by
the \ksmt rules, $n$ be the number of variables 
and
\[ \Lambda_i\coloneqq\{\alpha:\eqref{rule:extend}~\text{cannot be applied to}~(\alpha,\Lin_i,\NLin)~\text{linearly conflict-free}\}\text. \]
Then $\Lambda_i\supseteq\Lambda_{i+1}$ and $\Lambda_i\neq\Lambda_{i+n+2}$ hold for all $i$.
\end{lemma}
The proofs follow from the following:
\begin{enumerate}
\item \eqref{rule:extend} does not change
$\Lambda$ and can be applied consecutively at most $n$ times,
\item after application of \eqref{rule:resolve} or \eqref{rule:linearize}
the set $\Lambda$ is reduced which follows
from the properties of the resolvent, and \Cref{cor:lin-red-search-space} respectively, and
\item \eqref{rule:backtrack} does not change $\Lambda$ and can be applied only after \eqref{rule:resolve} or \eqref{rule:linearize}.
\end{enumerate}
\begin{corollary}
After at most $n+2$ steps
the search space is reduced.
\end{corollary}

\subsection{Concrete algorithm}\label{sec:conrete:alg}
%
The algorithm transforms the initial state by applying
\ksmt transition rules exhaustively. 
The rule applicability graph is shown in \Cref{fig:rule-transitions}.
The rule \eqref{rule:backtrack} is applicable whenever the linear part is false
in the current assignment. This is always the case after applications of  either \eqref{rule:resolve} or \eqref{rule:linearize}.
In order to check applicability of remaining rules 
\eqref{rule:extend},  \eqref{rule:resolve} and \eqref{rule:linearize} the
following conditions need to be checked. 
\begin{enumerate}
\item
    Is the state \notfalse? \label{comp:consistency} 
    In particular, we need to check whether the non-linear part evaluates to \False under the current assignment.
    Decidability of this problem for the broad class of functions $\Fdec$
    is shown in \Cref{decide:nonlinear}, along with concrete algorithms for common classes of non-linear functions.
\item\label{comp:conflict}
    If the state is \linnotfalse and a variable is chosen, can it be assigned in a way that the linear part remains conflict-free?
    A polynomial-time procedure is described
    in \Cref{sec:check-bounds}.
\end{enumerate}

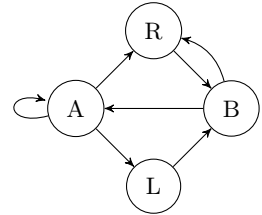
\begin{wrapfigure}[9]{R}{3.9cm}
\centering
\vspace{-2\baselineskip}
\begin{tikzpicture}[->,>=stealth',auto,node distance=4.5em]
\footnotesize
\tikzstyle{every node}=[draw,circle,inner sep=1.1ex]

\node (A)                    {\ref{rule:extend}};
\node (R) [above right of=A] {\ref{rule:resolve}};
\node (B) [below right of=R] {\ref{rule:backtrack}};
\node (L) [below right of=A] {\ref{rule:linearize}};

\path (B) 
          edge [bend right] (R)
          edge              (A)
      (R) edge              (B)
      (L) edge              (B)
      (A) edge [loop left]  (A)
          edge              (R)
          edge              (L)
;
\end{tikzpicture}
\vspace{-.7\baselineskip}
\caption{Transitions between applicability of rules.}
\label{fig:rule-transitions}
\end{wrapfigure}
These computations determine whether
\eqref{rule:extend}, \eqref{rule:resolve} or \eqref{rule:linearize} is applicable next.
\Cref{comp:conflict} has to be checked after each application of
\eqref{rule:extend} and \eqref{rule:backtrack}.
\label{info:backtracking-eq-backjumping}
Note that in case of transitioning to an application of rule \eqref{rule:backtrack} 
the size of the suffix of assignments to revoke
is syntactically available in form of the
highest position in $\alpha$ of a variable in $R_{\alpha,\Lin,z}$ or the
linearisation $L_{\alpha,\NLin}$, respectively.

Let us note that the calculus allows for flexibility in the choices of:
\begin{enumerate}
\item\label{it:heur:extend}
    The variable $z$ and value $q$ to assign to $z$ when applying rule \eqref{rule:extend}.
\item\label{it:heur:conflict}
    Which arithmetical resolutions to perform when applying rule \eqref{rule:resolve}.
\item\label{it:heur:linearize}
    Which linearisations to perform when applying rule \eqref{rule:linearize}.
    We describe the general conditions in \Cref{sec:nlin} and our approach in \Cref{impl:nonlinear}.
\end{enumerate}
Many of the heuristics presented in \cite{KorovinTV11,DBLP:conf/synasc/DraganKKV13} are applicable to \cref{it:heur:extend,it:heur:conflict} as well.

\subsection{Determining bounds and resolvents}\label{sec:check-bounds}

In this section we consider the problem of checking whether we can 
extend the trail of a linearly conflict-free state in such a way that the
linear part remains conflict-free after the extension and in this case we
apply rule \eqref{rule:extend}, or otherwise there is a conflict which should be
resolved by applying rule \eqref{rule:resolve}.

Given a \linnotfalse state $(\alpha,\Lin,\NLin)$ and a variable $z$ unassigned in $\alpha$, the problem
\[ \exists q\in\mathbb Q:\llbracket\Lin\rrbracket^{\alpha\cons z\mapsto q}\neq\False \]
can be solved efficiently 
by the following algorithm.
Let $\Lin_{z,\alpha}$ be those partially applied (by $\alpha$) clauses from $\Lin$ that only depend on $z$. The other clauses are either
already satisfied or depend on a further unassigned variable.
So each $D\in \Lin_{z,\alpha}$ is `univariate',  i.e. just a set of $z\diamond c_i$. The disjunction
of these simple predicates in $D$ is  equivalent to
a clause of the form (i) $z<a \vee z>b$, perhaps with non-strict inequalities,
giving an alternative between a lower and an upper bound or (ii)
a unit clause for a lower bound, or (iii)
a unit clause for an upper bound, or (iv) an arithmetic tautology.
So each clause is equivalent to the union of
at most two
half-bounded rational intervals. 
The conjunction of two such clauses corresponds to 
the intersection of sets of intervals, which is again 
a set of intervals. This intersection can be computed 
easily and can also
be checked for emptiness. In case the intersection is not empty, it even gives us intervals
to choose an assignment $q$ for $z$ with $\llbracket\Lin\rrbracket^{\alpha\cons z\mapsto q}\neq\False$.
If the intersection is empty,\pagebreak[2] we know there is no such $q$ and
we can use arithmetical resolution to resolve this conflict to obtain $R_{\alpha,\Lin,z}$.
\subsection{Non-linear predicates}\label{sec:nlin}
While resolution is a well-established and efficient symbolic technique
for dealing with the linear part of the CNF under consideration,
there seem to be no similarly easy techniques for non-linear predicates.
The approach presented here is based on numerical approximations instead.

Given a \linnotfalse state $(\alpha,\Lin,\NLin)$, in order to decide on the
applicability of 
\eqref{rule:linearize}, the
non-linear unit clauses in $\NLin$ have to be checked for validity under $\alpha$.
If all are valid, then, by definition, $(\alpha,\Lin,\NLin)$ is \notfalse.
\Cref{lem:pred-eval-computable} gives sufficient conditions on the non-linear 
functions in $\Fun_\nonlin$ in order to make this problem decidable.
In this section, we will describe how we deal with the case that 
some unit clause $\{P\}\in\NLin$ is \False under $\alpha$,
where according to \eqref{rule:linearize}
we construct a linearisation of $P$ with respect to $\alpha$.
We will not need the order of variables given in the trail $\alpha$,
so we will only use $\alpha$ as a partial assignment.
\begin{definition}\label{def:linearisation} Let $P$ be a non-linear predicate and let $\alpha$ be 
a partial assignment with $\llbracket P\rrbracket^\alpha=\False$.
An \emph{$(\alpha,P)$-linearisation} is a clause $L_{\alpha,P}= \{L_i : i\in I\}$  
consisting of finitely many rational linear predicates  $(L_i)_{i\in I}$  with the
properties
\begin{enumerate}
\item $\{\beta : \llbracket P\rrbracket^\beta=\True\}  \subseteq
\{\beta : \llbracket L_{\alpha,P}\rrbracket^\beta=\True\}$, and
\item $\llbracket L_{\alpha,P}\rrbracket^\alpha=\False$.
\end{enumerate}
\end{definition}
If we let $\vec c_\alpha$ denote the values assigned in $\alpha$ and
$\vec x$ the vector of assigned variables, we can reformulate the properties of
$ L_{\alpha,P}$ as a formula:
\[  \Big(P \implies \bigvee_{i\in I}  L_i \Big)\wedge\Big( \vec x = \vec c_\alpha \implies
\neg  \bigvee_{i\in I}  L_i \Big)
\]
This formula will not be added to the system but is just used as a basis for discussions. Later we will use 
a similar formalism to define linearisation clauses.

A central idea of our approach is to add $ L_{\alpha,P}$ as a new clause to the CNF, 
as well as the 
predicates $L_i$.
Adding  $ L_{\alpha,P}$  is sound, as the following lemma shows:
\begin{lemma}\label{lem:linearization-effect}
Suppose a partial assignment $\alpha$ violates a predicate $P$ with $\{P\}\in \NLin$, 
so $\llbracket P\rrbracket^\alpha=\False$. 
Further suppose $ L_{\alpha,P}$ is an $(\alpha,P)$-linearisation.
\begin{enumerate}
    \item Any $\beta$, which is a solution for $\Lin\cup\NLin$, is also 
    a solution for $\Lin\cup\{ L_{\alpha,P}\}\cup\NLin$.
    \item $(\alpha,\Lin\cup\{ L_{\alpha,P}\},\NLin)$ is not \linnotfalse.
\end{enumerate}
\end{lemma}
\begin{corollary}\label{cor:lin-red-search-space}
Whenever $\eqref{rule:linearize}$ is applied,
the search space is reduced.
\end{corollary}

Hence at least the partial assignment $\alpha$ (and all extensions thereof) are removed
from the search space for the linear part of our CNF, at the cost of
adding the clause $L_{\alpha,P}$ usually containing several new linear predicates $L_i$. In general, our linearisations will not just remove single points but rather polytopes from the search space.

We should emphasise several remarks on the linearisations:
There is a high degree of freedom when choosing a linearisation for a
given pair $(\alpha,P)$. Techniques for constructing these will be discussed
in \Cref{impl:nonlinear}. They will all be based
on numerical approximations.

Furthermore we are allowed to add more than one clause in one step, 
so we can construct several linearisations for different $(\alpha,P')$ 
as long as $\llbracket P'\rrbracket^\alpha=\False$, and then add all of them. 
This has already been formulated in \eqref{rule:linearize} as a \emph{set} of linearisation clauses $L_{\alpha,\NLin}$
instead of a single clause $L_{\alpha,P}$.

\section{Example}\label{example:complete}

As a basic example describing our method 
we consider the conjunction of the non-linear predicate $P:(x\leq\frac1y)$, 
and linear constraints $L_1: (x \geq y/4 +1) $ and $L_2: (x\leq 4\cdot(y-1))$, shown on Figure~\ref{fig:unsat-example}.
We will first detail on how linearisations can 
 be constructed 
numerically for $P$.
In \Cref{impl:nonlinear} we will detail on how linearisations can be constructed in general.

\paragraph{Linearisation of $P$.}
\label{linear:product}

\begin{figure}[t]
\vspace{-\baselineskip}
\includegraphics[height=.18\textheight]{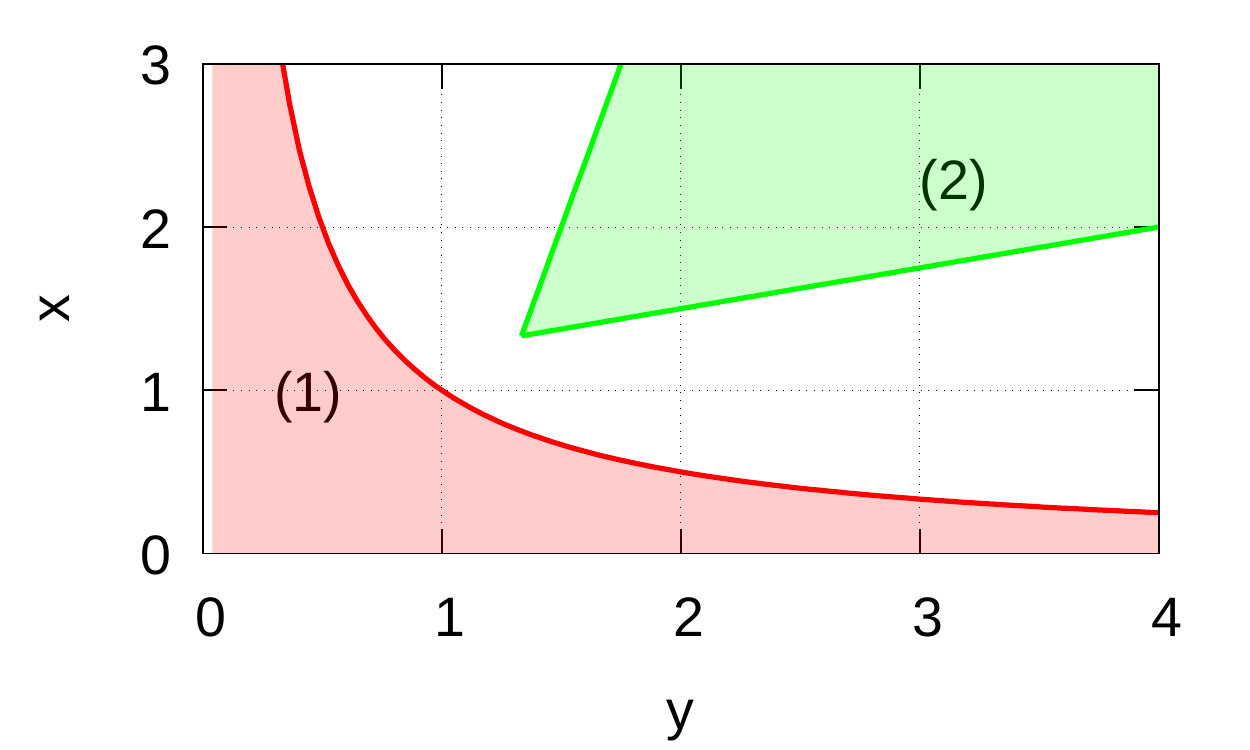}
\hfill
\includegraphics[height=.18\textheight]{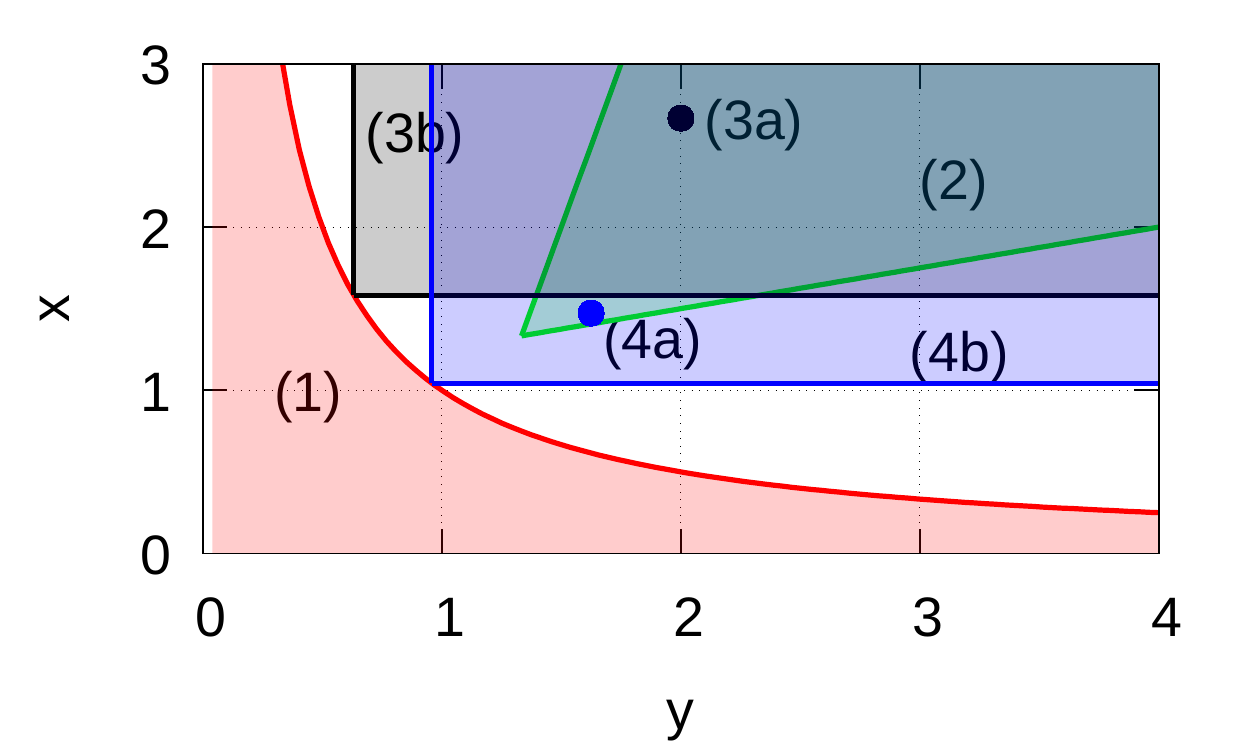}
\vspace{-\baselineskip}
\caption{Initial system and linearisations constructed.}
\label{fig:unsat-example}
\end{figure}

Assume $\llbracket P\rrbracket^\alpha=\False$ under assignment $\alpha$.
By definition, $\alpha$ assigns $(x,y)$ to some values $(c_x,c_y)$ such that
$c_x>\nicefrac1{c_y}$,  (point (3a), at $(\nicefrac83,2)$). Here we will only discuss the case $c_y>0$ needed below.
The other cases can be dealt with in a similar way.
To construct an $(\alpha,P)$-linearization, first we compute the rational number $d$ such that $1/c_y < d < c_x$.
In this example, we take $d\coloneqq(c_x+\nicefrac{1}{c_y})/2$, that is, for this linearisation $\nicefrac{19}{12}\approx1.58$.
%
In general, such values are computed by numerical approximations to the function value.
Then the clause $L_{\alpha,P}=\{ x \leq d, y \leq 1/d\}$ is the required linearisation (which excludes region (3b) containing the conflicting assignment).
Indeed,  $L_{\alpha,P}$ is implied by $P$ and  $\llbracket L_{\alpha,P}\rrbracket^\alpha=\False$.   

After adding $L_{\alpha,P}$ to the linear constraints, region (3b) is excluded from the search space
and backjumping to the empty assignment is performed
(since $\nicefrac83$ is not a linearly conflict-free assignment to $x$ anymore).
The system again is linearly conflict-free.
In the next iteration we obtain a solution (4a) roughly at $(1.47,1.63)$ to the new linear system,
linearisation at (4a) results in linear lemma excluding region (4b) where $d\approx1.04$.
Finally, the resulting linear constraints are unsatisfiable and therefore the
original system is proven to be also unsatisfiable.
This example is based on an actual run of our system.%
\section{Schemes for local linearisations}\label{sec:impl}

A successful linearisation scheme has to fulfil two tasks:
\hypertarget{it:nlin:dec}{(a)} deciding whether  a trail $\alpha$ is in conflict with a non-linear
predicate $P$ and then, if there is a conflict,
\hypertarget{it:nlin:linearize}{(b)} finding reasonable linearisations $L_{\alpha,P}$.
We first address task \hyperlink{it:nlin:dec}{(a)}.

\subsection{Deciding non-linear conflicts\label{decide:nonlinear}}

By  \Cref{def:separated}, $P$ is of the form $x\diamond f(\vec t)$, where
$f$ is a function symbol, $\vec t$ is a vector of terms, 
and $\diamond \in \{{<},{\leq},{>},{\geq}\}$. 
In the following assume that the terms in $\vec t$ 
use the variables $(y_1,\ldots,y_k)=\vec y\in V^k$. 
So the semantical interpretation $\llbracket f(\vec t)\rrbracket$ of the syntactical
term $f(\vec t)$ is a function $g:\R^k\to \R$.

In order to introduce the class $\Fdec$ we use the following notion of
approximable function.
\begin{definition}
We call a partial function $g:\mathbb R\to\mathbb R$ \emph{approximable} if
the set
\[ \square_g\coloneqq\{(p,q,s,t):g([p,q])\subset(s,t),\,p,q,s,t\in\mathbb Q\} \]
is computably enumerable.
Here, $g(I)$ denotes the set-evaluation of
$g$ on $I$, that is, $\{g(x):x\in I\cap\dom g\}$.
\end{definition}
This definition can easily be generalized to the multi-variate case by taking boxes
$[p_1,q_1]\times\cdots\times[p_k,q_k]$ with $\vec p,\vec q\in\mathbb Q^k$.
For 
total continuous real functions, approximability coincides with
the notion of computability known from Computable Analysis (
TTE) \cite{We00,BHW07}.

Given a number $d\in\Q$ and a vector $\vec c\in \Q^k$ with $d\neq g(\vec c)$
we can always decide whether $d \diamond g(\vec c)$ holds if $g:\mathbb R^k\to\mathbb R$ is
a total approximable function.
However, in general we cannot decide the premise $d\neq g(\vec c)$.
Therefore we restrict our considerations to a general class of functions where this problem is decidable.
\begin{definition}\label{def:F}
A partial function $g:\mathbb R^k\to\mathbb R$ is called a
\emph{function with decidable rational approximations}, denoted $g\in\Fdec$, if the following holds.
\begin{itemize}
\item $\dom(g)$ is decidable on $\Q^k$,
\item $\operatorname{graph}(g)$ is decidable on $\Q^k\times\Q$, and
\item $g$ is approximable.
\end{itemize}
\end{definition}

The following important classes of functions belong to $\Fdec$.

\begin{paragraph}{Multivariate polynomials.}
For  multivariate polynomials $g$ 
with rational coefficients, rational arguments are mapped to rational results using 
rational arithmetic and the relations $\diamond$ under consideration are decidable on $\mathbb Q^2$.
%
\end{paragraph}

\begin{paragraph}{Selected elementary transcendental functions.} 
Let $g\in\{\exp,\linebreak[1]\ln,\linebreak[1]\log_b,\linebreak[1]\sin,\linebreak[1]\cos,\linebreak[1]\tan,\linebreak[1]\arctan\}$, where in the case of $\log_b$, $b\in\mathbb Q$. 
Let us show that $g \in \Fdec$.
Indeed, it is well known that $g:\mathbb R\to\mathbb R$ is computable~\cite{We00}.
Since emptiness of $[p,q]\setminus\dom g$ is decidable, $g$ is also approximable.
In addition, $X_g\coloneqq\operatorname{graph}(g)\cap\mathbb Q^2$ either consists of
a single point, or in the case of $\log_b$, is of the form $X_g=\{(b^n,n):n\in\mathbb Z\}$~\cite{niven_1956} and therefore is decidable, as is the respective domain.
\end{paragraph}

\begin{paragraph}{%
Selected discontinuous functions.}
Additionally, $\Fdec$ includes some discontinuous functions like e.g.\ the step-functions taking rational values with discontinuities at finitely many rational points
and more generally piecewise polynomials defined over intervals with a decidable set of rational endpoints.
Multi-variate piecewise defined functions with
non-axis-aligned discontinuities are included as well.
\end{paragraph}

%
\begin{lemma}\label{lem:pred-eval-computable}
Let $P$ be a predicate over reals 
and let $\alpha$ be a trail assigning 
all variables used in $P$.
If $P$ is linear or $P:(x\diamond f(\vec t))$ 
with $\llbracket f(\vec t)\rrbracket\in\Fdec$ then
$\llbracket P\rrbracket^\alpha$ is computable.
\end{lemma}
\pagebreak[2]
\begin{proof}
By definition, trails $\alpha$ contain rational assignments.
If $P$ is linear, there is nothing to show.
Let $P:(x\diamond f(\vec t))$ with $g(\vec y)=\llbracket f(\vec t)\rrbracket\in\Fdec$
where $\vec y$ is the vector of free variables in terms $\vec t$.
The cases $\llbracket\vec y\rrbracket^\alpha\notin\dom g$ and
$\llbracket(\vec y,x)\rrbracket^\alpha\in\operatorname{graph}(g)$ are decidable
by the definition of $\Fdec$.
	The remaining case is $\vec z\coloneqq\llbracket\vec y\rrbracket^\alpha\in\dom g$
    and $\llbracket(\vec y,x)\rrbracket^\alpha\notin\operatorname{graph}(g)$.
	Perform a parallel search for
	1)~$q\in\Q$ with $(\vec z,q)\in\operatorname{graph}(g)$ and for
	2)~a rational interval box $\vec I\times J$ in $\square_{\tilde g}$ with
	$\vec z\in\vec I$ and $\llbracket x\rrbracket^\alpha\notin J$.
	We now show that this search terminates.
	Either $g(\vec z)\in\mathbb Q$, then $q=g(\vec z)$ can be found in the
	graph of $g$, or
	$g(\vec z)\notin\mathbb Q$, then 
	$|\llbracket x\rrbracket^\alpha-g(\vec z)|>0$, thus
	there is a rational interval box $\vec I\times(s,t)\in\square_g$
	with $\vec z\in\vec I$ and $s,t\in\Q$
	such that $\llbracket x\rrbracket^\alpha\notin(s,t)$.
	Note that $\vec I$ can be the point-interval $[\vec z]$ since
	$\vec z\in\dom g$.
\end{proof}


In particular, if all predicates $P:(x\diamond f(\vec t))$ appearing in a given
problem instance are such that the function $\llbracket f(\vec t)\rrbracket$
used in this instance are from $\Fdec$, we can decide if a \ksmt state is
conflict-free as required in \Cref{sec:conrete:alg}.



\subsection{Linearisations for functions in $\Fdec$}\label{impl:nonlinear}

This section addresses task \hyperlink{it:nlin:linearize}{(b)},
namely  finding reasonable linearisations $L_{\alpha,P}$ in case a trail $\alpha$ 
is in conflict with a non-linear predicate $P$, that is,
$\llbracket P\rrbracket^\alpha=\False$.
In order to reduce the number of cases, we assume that the comparison operator $\diamond$ in $P:x\diamond f(\vec t)$ is from
$\{<,\leq\}$. The other two cases $\{>,\geq\}$ are symmetric.
 
Again let $g=\llbracket f(\vec t)\rrbracket:\R^k\to\R$ be the function represented by the term $f(\vec t)$. We assume that $g \in \Fdec$.
Let $c_x=\llbracket x\rrbracket^\alpha\in \Q$
and
$\vec c_{\vec y}=\llbracket \vec y\rrbracket^\alpha\in\Q^k$ 
be the values assigned by $\alpha$ to the free variables $\vec y$ in $\vec t$,
additionally, let $\vec c_{\vec y}\in\dom g$.
Furthermore let
$c_g=g(\vec c_{\vec y})=\llbracket f(\vec t)\rrbracket^\alpha\in \R$ be the value resulting from an 
exact evaluation of $g$.
Note that $c_g$ will only be used in the discussion and will not be added to the constraints, since in general $c_g\notin\mathbb Q$.
Then our assumption of an existing conflict 
$\llbracket P\rrbracket^\alpha=\False$ can be read as 
$c_x > c_g$ for $\diamond\in\{{\leq}\}$,
and as $c_x \geq c_g$ for $\diamond\in\{{<}\}$.
Let us note that $c_x$ and
$\vec c_{\vec y}$ are rational, but $c_g$ is a real number and usually irrational.
Since $g\in \Fdec$ 
we can compute approximations $\bar c_g\in\Q$ to $c_g$ with $|\bar c_g - c_g| \leq \varepsilon$ for any rational $\varepsilon>0$ using
\Cref{lem:pred-eval-computable}.

\pagebreak[2]

We now give a list of possible linearisations of $g$,
starting from trivial versions where we exclude just the conflicting point
$(c_x, \vec c_{\vec y})$ to more general linearisations excluding larger regions
containing this point.
\todofb{`special properties of $g$' are more explicit, helps understanding when which lin.\ is applicable.}
%
\begin{description}
\item[Point Linearisation:]
A trivial $(\alpha,P)$-linearisation excluding the point $(c_x, \vec c_{\vec y})$ is
\[(\vec y = \vec c_{\vec y}\implies x \not= c_x)\]
\item[Half-Line Linearisation:]
An $(\alpha,P)$-linearisation excluding a closed half-line starting in $c_x$ is
\[(\vec y = \vec c_{\vec y}\implies x < c_x)\]
\end{description}
In the following we will develop more powerful linearisations with the aim to exclude larger regions of the search space. 

For better linearisations, we can exploit additional
information about the predicate $P$ and the trail $\alpha$, especially about the
behaviour of $g$ in a region around $\vec c_{\vec y}$.
This information could be obtained by a per-case analysis on $\Fun_{\nonlin}$,
or during run time using external algebra systems or 
libraries for  exact real arithmetic or interval arithmetic
on the extended real numbers $\R\cup\{-\infty,+\infty\}$.
Our focus, however, is on the numerical and not the symbolical approach.

As we aim at linearisations, the regions should
have linear rational boundaries, so we concentrate on finite intersections of half-spaces:
%
%
\begin{definition}
An (open or closed) \emph{rational half-space} $H\subseteq \R^k$ is the solution set of a linear predicate $\vec a\cdot \vec y \leq b$ or  $\vec a\cdot \vec y < b$ for some $\vec a\in\mathbb Q^{k}$, $ b\in\mathbb Q$.
A \emph{rational polytope} $R\subseteq \R^k$ is a finite intersection of rational half-spaces.
\end{definition}
Any such polytope $R$ is a convex and possibly unbounded set and can be
represented as the conjunction of linear predicates over the variables $\vec y$. 
Therefore the complement $\R^k\setminus R$ can be represented as a linear clause $\{L_i:i\in I\}$ denoting
the predicate $\vec y\notin R$. For the ease of reading, instead of writing clauses like $\bigvee_{i\in I}L_i\lor D$ we will use
$\vec y\in R\implies D$ in the following.

Since $g\in \Fdec$ and approximable it follows that  for
any bounded rational polytope $R\subseteq \R^k$ 
in the domain of $g$ we can find 
arbitrarily precise rational over-approximations $(a,b)$ such that $g(R)\subset(a,b)$.


\begin{description}
\item[Interval Linearisation:]
Suppose we have $c_x\neq c_g$. By approximating $c_g$ we compute
$d\in\Q$ with $c_g < d < c_x$.
The proof of \Cref{lem:pred-eval-computable} provides an initial rational polytope
$R\in \R^k$  with $\vec c_{\vec y}\in R$
such that $d\not\in g(R)$.
Then
 \begin{equation}\label{eq:intlin}\small
 \vec y\in R\implies x\leq  d\
\end{equation}
is an $(\alpha,P)$-linearisation.
Using specific properties of $g$, e.g., monotonicity, we can extend the polytope
$R$ to an unbounded one.
\end{description}
This linearisation excludes the set $\{x: x> d\}\times R$
from the search space which is a polytope, now in $\R\times\R^k$,  containing the point $(c_x,\vec c_{\vec y})$.
\pagebreak[3]

Linearisations in Example~\ref{example:complete} are  of this type,
there $c_g=1/c_y$ and $R=(1/d,\infty)$  defined by $y>1/d$ which is the negation
of $y\leq 1/d$, the second literal in the linear lemma $L_{\alpha,P}$ is the
right hand side of the implication~\eqref{eq:intlin}.


The univariate predicate $x\leq d$ corresponds to a very special
half-space in $\R\times \R^k$, as it is independent from the variables in $\vec y$.
Usually, using partial derivatives gives better linearisations: 

\begin{description}
\item[Tangent Space Linearisation:]
    Suppose we again have $c_x\neq c_g$. 
    Assume the partial derivatives of $g$ at $\vec c_{\vec y}$ exist
    and we are able to compute a vector $\vec c_{\partial{}}=(c_1,\ldots,c_k)$
    of rational approximations, that is,
    $c_i\approx \frac{\partial g}{\partial y_i}(\vec c_{\vec y})$.
    As before we construct $d\in \Q$ 
    with $c_g < d < c_x$
    and search for a rational polytope $R\in\R^k$
    with $\vec c_{\vec y}\in R$.
    But instead of just $d\not\in g(R)$ now $R$ has to fulfil the constraint 
    \[
    \forall \vec r \in R:  g(\vec r) \leq  d+ \vec c_{\partial{}}\cdot (\vec r-\vec c_{\vec y})
    \]
    using the dot product of $\vec c_{\partial{}}$ and $(\vec r-\vec c_{\vec y})$.
    Again, $R$ can be found using \todo{maybe to refer to \texttt{tangentspace}?}
    interval computation. 
    Then
    \[\vec y\in R\implies x \leq d+ \vec c_{\partial{}}\cdot (\vec y-\vec c_{\vec y})\] 
    is an $(\alpha,P)$-linearisation, since the dot-product is a linear and rational
    operation.
    This situation is schematically depicted in \Cref{fig:PDL}.
\end{description}
\begin{figure}[t]
\centering
\includegraphics{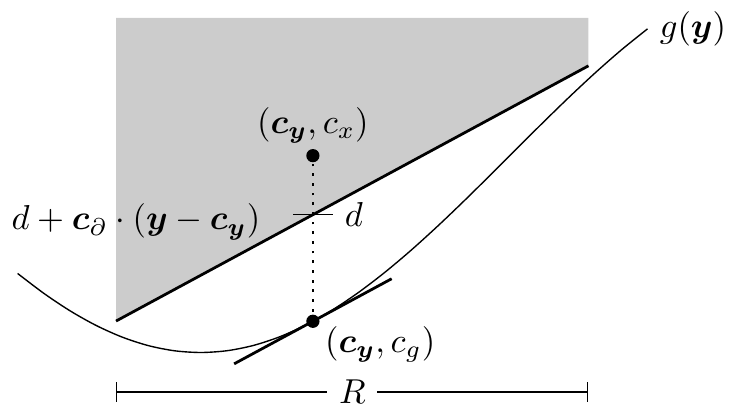}
\caption{Tangent Space Linearisation shown for univariate $g$. The shaded area will be excluded from the search space.}
\label{fig:PDL}
\end{figure}
Using the tangent space, we are able to get a much better `fit' of
 $d+\sum c_iy_i$ to $g$ than just using the naive interval evaluations.
This allows to choose $d$ closer to $c_g$ for given $R$, or to choose a bigger polytope $R$ for a given $d$.
Some examples of Tangent Space Linearisations are available.\cref{footnote:ksmt}

\begin{lemma}
By construction, the above procedures indeed provide linearisations as stated in
\Cref{def:linearisation}.
\end{lemma}

For the rest of this section, we briefly discuss more specific linearisations for some important cases
when we can perform a by-case analysis on $g$
and exploit further properties like (piecewise) monotonicity, convexity
or boundedness, which cannot be deduced by naive interval arithmetic,
see \Cref{sec:eval} for details. 
%
%
\pagebreak[2]
\begin{itemize}
\item $g(y)=y^{2n}$ is convex, with polytope $R=(-\infty,+\infty)$ for $\diamond\in\{{>},{\geq}\}$. 
\item $g(y)=y^{2n+1}$ is monotonically increasing, with polytopes $R$ of the form $(-\infty,c]$, 
similar to the linearisation in \Cref{linear:product}.
\item Polynomials can be decomposed into monomials.
\item Piecewise convex/concave functions $g$ like $\sin,\cos,\tan$ allow polytopes covering a 
convex area in their domain.
\item  More direct ways of computing linearisations of the elementary transcendental functions 
can be obtained e.g.\ by bounding the distance of the image of specific $g$ 
to algebraic numbers, 
such bounds are given in \cite[section 4.3]{lefev2000}.
\end{itemize}

\section{Evaluation}\label{sec:eval}
We implemented our approach in the \ksmt system, which is open source and publicly available~\footnote{\label{footnote:ksmt}\url{http://informatik.uni-trier.de/~brausse/ksmt}}.
The \ksmt system supports a subset of \texttt{QF\_LRA} and \texttt{QF\_NRA} logics as defined in the SMT-LIB standard.
As with \texttt{Z3}, when no logic is specified in the input script, our extended signature $\mathbb R_\nl$ is the default.

Choices made in the implementation include:
\begin{itemize}
\item
    Selecting a rational value in a non-empty interval as smallest dyadic or by continued fractions.
\item
    The decision which clauses to resolve on conflict is guided by an internal SAT-solver.
\item
    Heuristic about reusing existing constraints when computing polytope $R$,
    leading to piecewise linear approximations of $g$.
\item
    Specialised linearisation algorithms for specific combinations of subclasses of functions $g\in\Fdec$ and $\vec c_{\vec y}$:
    \begin{description}
    \item[differentiable:]
        Use Tangent Space Linearisation.
    \item[convex/concave:]
        Derive the polytope $R$ from computability of unique intersections between $g$ and the linear bound on $\vec y$.
    \item[piecewise:]
        This is a 
        meta-class in the sense that
        %
        %
        $\dom g$ is partitioned into $(P_i)_{i\in I}$ where the $P_i$ are linear or non-linear predicates in $\vec y$, and for each $i\in I$ there is a linearisation algorithm, then the decision which linearisation to use is based on membership of the exact rational value $\vec c_{\vec y}$ in one of the $P_i$.
    \item[rational:] 
        Evaluate $c_g$ exactly in order to decide on linearisation to use.
   \item[transcendental:] 
        Bound $|c_x-c_g|$ by a rational from below by approximating $c_g$ by the TTE implementation \texttt{iRRAM}\footnote{\url{http://irram.uni-trier.de}}
        \cite{Muller00}
        in order to compute $d$.
    \end{description}

\end{itemize}


We evaluated our approach over higher dimensional sphere packing benchmarks which are available at~\cref{footnote:ksmt}.
Sphere packing is a well known problem which goes back to Kepler's conjecture,
and in higher dimensions is also of practical
importance e.g., in error correcting codes.
The purpose of this evaluation is to exemplify 
that our approach is viable 
and can contribute to the current state-of-the-art, extensive evaluation is left for future work.
\begin{table}[t]
\begin{center}
\begin{tabular}[c]{c@{: }p{4.5em}}
	s & `\SAT' \\
	$\delta$ & '$\delta$-\texttt{sat}', $\delta=10^{-3}$ \\
	u & `\UNSAT' \\
	? & `\UNKNOWN' \\
	> & timeout
\end{tabular}%
\;\;\;%
\begin{tabular}[c]{cc|lr|lr|lr|lr|lr|lr|lr}
	$d$ & $n$
	& \multicolumn{2}{c|}{\texttt{ksmt}}
	& \multicolumn{2}{c|}{\texttt{cvc4}}
	& \multicolumn{2}{c|}{\texttt{z3}}
	& \multicolumn{2}{c|}{\texttt{mathsat}}
	& \multicolumn{2}{c|}{\texttt{yices}}
	& \multicolumn{2}{c|}{\texttt{dreal}}
	& \multicolumn{2}{c}{\texttt{rasat}}
	\\ \hline
	\multirow{5}{*}{$2$} & $2$ & s & 0.01s & ? & 0.03s    & s & 0.01s    & s & 0.02s & s & 0.01s & $\delta$ &    0.01 & s & 0.02 \\
	                     & $3$ & s & 0.03s & ? & 0.08s    & > & 60m      & s & 0.24s & s & 0.03s & $\delta$ &    0.02 & > & 8h   \\
	                     & $4$ & > &  8h   & u & 1474.16s & > & 60m      & u & 8.11s & > & 17h   & $\delta$ &    0.05 & > & 8h   \\
	                     & $5$ & u & 1.43s & u & 0.45s    & > & 8h       & u & 0.28s & > & 8h    & u        & 3581.96 & > & 8h   \\
	                     & $6$ & u & 5.00s & u & 0.75s    & > & 8h       & u & 0.40s & > & 166m  & >        &    8h   & > & 8h   \\ \hline
	\multirow{2}{*}{$3$} & $5$ & s & 0.93s & ? & 465.45s  & > & 8h       & s & 0.12s & s & 0.06s & >        &    8h   & > & 8h   \\
	                     & $6$ & s & 6.02s & > & 143m     & > & 7h       & > & 8h    & > & 6h    & >        &    8h   & > & 8h   \\ \hline
	\multirow{3}{*}{$4$} & $5$ & s & 0.38s & ? & 1544.87s & s & 2165.78s & s & 0.10s & s & 7.34s & >        &    8h   & > & 8h   \\
	                     & $6$ & s & 0.57s & > & 91m      & > & 8h       & s & 0.23s & s & 0.38s & >        &    8h   & > & 8h   \\
	                     & $7$ & s &14.27s & > & 160m     & > & 8h       & s & 0.18s & > & 8h    & >        &    8h   & > & 8h
\end{tabular}%
\end{center}
\caption{Benchmarks of $K_{n,d}$ for different $n,d$.}
\label{tab:bench-K}
\end{table}
\begin{table}[t]
\begin{center}
\begin{tabular}[c]{c|lr|lr|lr|lr|lr|lr|lr}
    $r$ 
	& \multicolumn{2}{c|}{\texttt{ksmt}}
	& \multicolumn{2}{c|}{\texttt{cvc4}}
	& \multicolumn{2}{c|}{\texttt{z3}}
	& \multicolumn{2}{c|}{\texttt{mathsat}}
	& \multicolumn{2}{c|}{\texttt{yices}}
	& \multicolumn{2}{c|}{\texttt{dReal}}
	& \multicolumn{2}{c}{\texttt{raSAT}}
	\\ \hline
    $\sqrt{37}$ & u &  0.07s & u & 0.76s & u &   510.67s & u &   40.55s & u &  0.07s & u        &  0.01 & > & 8h \\
    $\sqrt{49}$ & u &  0.40s & u & 2.46s & u & 23211.20s & u & 6307.18s & u &  0.11s & u        &  0.03 & > & 8h \\
    $\sqrt{62}$ & u & 11.61s & u & 5.07s & u &   210.16s & > &    14.5h & u & 76.82s & u        &  2.00 & > & 8h \\
    $\sqrt{63}$ & u & 55.84s & ? & 0.48s & u &  3925.65s & > &    14.5h & u &  0.10s & u        & 12.38 & > & 8h \\
    $\sqrt{64}$ & s &  0.01s & ? & 0.01s & s &     0.00s & > &    21.6h & s &  0.00s & $\delta$ &  0.01 & > & 8h
\end{tabular}
\end{center}
\caption{Benchmarks of $C_r$ for different $r$.}
\label{tab:bench-B}
\end{table}

The solvers\footnote{%
    \texttt{\ksmt-0.1.3}, \texttt{cvc4-1.6+gmp}, \texttt{z3-4.7.1+gmp}, \texttt{mathsat-5.5.2}, \texttt{yices-2.6+lpoly-1.7}, \texttt{dreal-v3.16.08.01}, \texttt{rasat-0.3}}
were compiled with GCC-8.2 according to their respective documentation (except for \texttt{mathsat}, which is not open-source). 
Experiments were run on a machine with 32~GiB RAM, 3.6~GHz Core i7 processor and Linux 3.18.

\begin{example}[Sphere packing] 
Let $n,d\in\mathbb N$ and let
\[ K_{n,d}\coloneqq\exists\vec x_1,\ldots,\vec x_n\in\mathbb R^d:
	\bigwedge_{1\leq i\leq n}\Vert\vec x_i\Vert_\infty\leq 1\land
	\bigwedge_{1\leq i<j\leq n}\Vert\vec x_i-\vec x_j\Vert_2>2 \]
An instance $K_{n,d}$ is \SAT iff $n$ balls fit into a $d$-dimensional box of
radius $2$ without touching each other.
In the SMT-Lib language the $\Vert\cdot\Vert_\infty$ norms in these instances
are formulated using per-component comparisons to the lower and upper endpoints
of the range, while the euclidean norms $\Vert\vec s\Vert_2>t$ are expressed by
the equivalent squared variant $\sum_i\vec s_i^2>t^2$.
\Cref{tab:bench-K} provides a comparison of
different solvers on instances of this kind.
\end{example}
\begin{example} 
Let $r\in\mathbb Q$, then
\[ C_r\coloneqq\exists\vec x,\vec y\in\mathbb R^3:\Vert\vec x\Vert_2^2\leq r^2\land\Vert\vec y\Vert_2^2\geq 8^2\land\Vert\vec x-\vec y\Vert_\infty\leq\tfrac1{100}\text. \]
$C_r$ is \SAT for some $r\in[0,8]$ iff there is a translation of the center
$\vec x$ of the 3-dimensional ball $B_r(\vec x)$ in a box of radius
$\frac1{100}$ such that it intersects the complement of $B_8(\vec y)$.
Since the constraints are expressed as square-root-free expressions, obviously for $r\geq 8-\frac1{100}$, there is a solution.
\Cref{tab:bench-B} list running times for various $r$ of our solver and other solvers of non-linear real arithmetic.

Noteworthy about these benchmarks is the monotonicity of the running times of \ksmt in contrast to e.g.\ \texttt{yices} in conjunction with unlimited precision, which seems to be what prevents \texttt{cvc4} from deciding the instance for $r=\sqrt{63}$ and even $r=\sqrt{64}$.
\end{example}



These experiments show that already in the early stage of the implementation, our system can handle high dimensional non-linear problems which are challenging for most SMT solvers.

\section{Conclusions and future work}
In this paper we presented a new approach for solving non-linear constraints over the reals. 
Our \ksmt calculus  
combines model-guided solution search with targeted linearisations for resolving non-linear conflicts. 
We implemented our approach in the \ksmt system, our preliminary evaluation shows promising results demonstrating viability of the proposed approach. 

For future work we are developing more precise linearisations for specific trigonometric functions
and are analyzing the complexity of deciding conflicts in general.
We are working on extending the applicability of our implementation and a more extensive evaluation.
We are also investigating theoretical properties of our calculus, such completeness in restricted settings and $\delta$-completeness.
 

\paragraph{Acknowledgements.}
We thank the anonymous reviewers and Stefan Ratschan for their helpful comments.

\bibliographystyle{abbrv}

{\tiny
\bibliography{main}
}

\end{document}